\documentclass[pra,floatfix,amsmath,twocolumn,showpacs,superscriptaddress]{revtex4-1}
\usepackage[ansinew]{inputenc}

\usepackage{verbatim}
\usepackage{epsfig}
\usepackage{pst-all}
\usepackage{color}
\usepackage{dcolumn}
\usepackage{amsmath}
\usepackage{amsthm}
\usepackage{bm}
\usepackage{layout}
\usepackage{float}
\usepackage{txfonts}
\usepackage{amsfonts}
\usepackage{amssymb}%
\usepackage[ansinew]{inputenc}
\usepackage{amsmath}
\usepackage{amsthm}
\usepackage{float}
\usepackage{amsfonts}
\usepackage{amssymb}

\newtheorem{lemma}{Lemma}
\newtheorem{thm}[lemma]{Theorem}

\begin{document}

\title{Monogamy deficit for quantum correlations in multipartite quantum system}

\author{Si-Yuan Liu}
\affiliation{Institute of Modern Physics, Northwest University, Xian
710069,  China }
\affiliation{Institute of Physics, Chinese Academy of Sciences, Beijing 100190, China}
\author{Bo Li}
\affiliation{Institute of Physics, Chinese Academy of Sciences, Beijing 100190, China}
\author{Wen-Li Yang }
\affiliation{Institute of Modern Physics, Northwest University, Xian
710069,  China }
\author{Heng Fan}
\email{hfan@iphy.ac.cn}
\affiliation{Institute of Physics, Chinese Academy of Sciences, Beijing 100190, China}

\date{\today}

\begin{abstract}
We introduce the concept of monogamy deficit for quantum correlation by combining together
two types of monogamy inequalities depending on different measurement sides.
For tripartite pure state, we demonstrate a relation which connects two types of monogamy inequalities
for quantum discord and provide the difference between them. 
By using this relation, we obtain an unified physical interpretation for these two monogamy deficit.
 In addition,  we find an interesting fact that there is a general monogamy condition
 for several quantum correlations for tripartite pure states.
 We then provide a necessary and sufficient condition for the establishment of
 one kind of monogamy inequality for tripartite mixed state
  and generalize it to multipartite quantum state.
\end{abstract}


\maketitle

\section{Introduction}
Quantum correlations, such as entanglement and quantum discord,
are assumed to be resources in quantum information processing and are different from classical
correlations. On the other hand, in general, entanglement and discord
are different from each other.
Previous studies focus on entanglement which
is a special quantum correlation enabling fascinating quantum information tasks such as
super-dense coding\cite{CH}, teleportation\cite{CH2}, quantum cryptography\cite{AK},
remote-state preparation\cite{AK2} and so on. However,
some quantum applications superior than their classical counterparts are found
with vanishing or negligible entanglement \cite{CH3,EK,AD}.
In this sense, entanglement seems not
capture all the quantum features of quantum correlations.
So other measures of quantum correlations are proposed. Among those measures
that in general go beyond entanglement, quantum
discord is a widely accepted one in recent years
\cite{Zurek,Vedral,key-15}. The analytic results of quantum discord
and its physical meaning are studied extensively, for example, in Refs. \cite{AD,PG,AS,AS2,DC}.
The experiments about quantum discord are implemented \cite{RA,RA2}.
Quantum discord can also be generalized to the multipartite
situation \cite{WO,KM,IC,ccm11}, for more results, see a recent review paper \cite{key-15}.
%Next, we generally refer quantum correlation to quantum discord since its definition
%unless specified particularly.

There are many fundamental differences between classical correlation
and quantum correlations. One of them is  the shareability of correlation  among
many parties. Generally speaking, classical correlation can be freely
shared among many parties, while quantum ones do not have this property.
For example, for tripartite pure state, if two parties are highly entangled,
they cannot have a large amount of entanglement shared with a
third one. The limits on the shareability of quantum correlations
are described by monogamy inequalities. Much progresses have already been made
about the monogamy properties of various quantum correlations \cite{JO1,KOa,SA,key-12,key-13,key-14}.
As one application,
the monogamy property of quantum correlations also play a fundamental role for
the security of the quantum key
distributions\cite{LoChau,MP}. Some known monogamy properties of entanglement
measure are, for example, concurrence and squashed entanglement \cite{FE,key-1,key-14,KO}.

It is shown that the monogamy relation does not always be satisfied by the quantum correlations \cite{JO1}.
So it is necessary to know when a specified quantum correlation can satisfy this property.
Concerning about quantum discord, in general, it does
not satisfy this nature \cite{MA}.
However, it may have some interesting applications in
case the monogamy condition is satisfied \cite{ma,hufan}.
We should note that there are two types
of monogamy inequalities for quantum discord since it is asymmetric depending on
the measurement side for a bipartite state \cite{XI}.
A necessary and sufficient condition
for one type of monogamy relation satisfying is given where only
one side of measurement is studied \cite{MA}. A natural question is then
that does there exist an analogous property for another class of monogamy with measurement taken
on a different side?  In this paper,  we first demonstrate
a relation between those two types of monogamy conditions for tripartite pure state and provide the difference between them. By using this relation, 
we provide an unified physical interpretation for these two monogamy deficit and generalize it to the $N$-partite pure state.
Then we give a necessary and
sufficient condition for the holding of the second type of monogamy relation
and  further generalize the result to the $N$-partite system.
In particular, those two types of monogamy relations are generally studied independently.
Our result that two monogamy inequalities can be combined together by introducing monogamy
deficit provides a new, in general, more complete viewpoint. This can enlighten
much research both on quantum correlation and monogamy property.

\iffalse
This paper is organized as follows.  In the next section, after
introducing the definition of monogamy deficit, we derive a
relation between these two types of monogamy deficit for tripartite pure state.
Using thisn unified physical interpretation for these two monogamy deficit.
In addition,  we find an interesting fact that there is a general monogamy condition for several quantum correlations for tripartite pure states.
Then, we provide a necessary
and sufficient condition for the non-negative second monogamy deficit for tripartite
mixed state. We also use a figure to illustrate the monogamy deficit under specific
pure state. By considering the equivalent expression of the necessary
and sufficient condition for tripartite pure state, we obtain a new
physical interpretation of entanglement of formation. Finally,
as an further application of the above necessary and sufficient condition,
we extend our result to the $N$ partite system.
\fi

\section{The monogamy deficit for pure state}\label{connection1}

\subsection{The connection of two types of monogamy deficit}\label{connection}
Quantum discord is defined as the difference between mutual information, which
is accepted to be the total correlation,
and maximum classical mutual information \cite{Zurek,Vedral}
\begin{eqnarray}
D^{\rightarrow }(\rho _{AB}) &=&\widetilde{I}(\rho _{AB})-I^{\rightarrow }(\rho _{AB})
\notag \\
&=&S(\rho _{A})-S(\rho _{AB})+\min_{\{\mathcal{M} _{i}\}}\sum_{i}p_{i}S(\rho _{B|i}),
\label{discord}
\end{eqnarray}
where arrow ``$\rightarrow$'' means measurement on $`A'$ and ``$\leftarrow$'' means
measurement on $`B'$, ${{M} _{i}}$ represent POVM measurement performed on $A$
for a bipartite state $\rho _{AB}$. So quantum discord is considered describing the quantumness
of correlations.
In this paper, we mainly use the same notations as those in Ref.\cite{MA}.
We use the left arrow (``$\leftarrow$'') and the right arrow (``$\rightarrow$'')
to distinguish the side of the measurement. Also we have notations, $\widetilde{I}(\rho _{AB})=S(\rho _{A})-\widetilde{S}(\rho _{A|B})$, and
$\widetilde{S}(\rho _{A|B})=S(\rho _{AB})-S(\rho _{B})$, here $S(\sigma)=-{\rm tr}(\sigma\log_2\sigma)$ is the von Neumann entropy of
a density matrix $\sigma$.
By those definitions presented above, we know that quantum discord in general should be
asymmetric and depend on the measurement side, which can be either $A$ or $B$.
It is understandable that those two definitions possess different fundamental properties.

Recently, two kinds
of monogamy inequalities have been studied in Refs.\cite{MA,giorgi} and  \cite{XI}.
For a tripartite state $\rho _{ABC}$,
by combining two monogamy inequalities together, we define two kinds of monogamy deficit of quantum discord,
\begin{eqnarray}
\bigtriangleup_{D_A}^{\leftarrow }=D^{\leftarrow }(\rho_{A|BC})-D^{\leftarrow }(\rho_{AB})-D^{\leftarrow }(\rho_{AC}),
\label{monogam1}
\end{eqnarray}
\begin{eqnarray}
\bigtriangleup_{D_A}^{\rightarrow }=D^{\rightarrow }(\rho_{A|BC})-D^{\rightarrow }(\rho_{AB})-D^{\rightarrow }(\rho_{AC}).
\label{monogamyde}
\end{eqnarray}
It is worth noting that a similar quantity as defined in
Eq. (2) and (3) has also been introduced by Bera \emph{et al.} \cite{MN}.
In Eq. (\ref{monogam1}), the first term $D^{\leftarrow }(\rho_{A|BC})$
involves a positive operator valued measurement (POVM) performed on $B$ and $C$,
and the other involve measurements
only on $A$.

Because of the asymmetry of quantum discord, the above two monogamy deficit are apparently quite different.
In this paper, however, we find that there is a relation between them, which means that the
monogamy relations on one of them provide some limits on another.
We first have a following observation.
For two kinds of monogamy deficit $\bigtriangleup_{D_A}^{\leftarrow }$, $\bigtriangleup_{D_A}^{\rightarrow }$
of an arbitrary tripartite pure state $\rho _{ABC}$, we find,
 \begin{eqnarray}
\bigtriangleup_{D_A}^{\leftarrow }&=&\frac{1}{2}(\bigtriangleup_{D_B}^{\rightarrow }+\bigtriangleup_{D_C}^{\rightarrow }),
\label{1deficit} \\
\bigtriangleup_{D_A}^{\rightarrow }&=&\bigtriangleup_{D_B}^{\leftarrow }+\bigtriangleup_{D_C}^{\leftarrow }-\bigtriangleup_{D_A}^{\leftarrow }.
\label{relation}
\end{eqnarray}
The proof of those two relations can be the following.
For simplicity, denote $S(\rho_{B|A})$ as the optimal conditional entropy of $I^{\rightarrow }(\rho _{AB})$ after the measurement which defined as $\min_{\{\mathcal{M} _{i}\}}\sum_{i}p_{i}S(\rho _{B|i})$, that is  $I^{\rightarrow }(\rho _{AB})=S(\rho _{B})-\min_{\{\mathcal{M} _{i}\}}\sum_{i}p_{i}S(\rho _{B|i})=S(\rho _{B})-S(\rho_{B|A})$.
Using the Koashi-Winter formula \cite{KO}, we have $S(\rho_{B|A})=E(\rho_{BC})$, where $E(\rho_{BC})$ means the entanglement of formation for a bipartite state $\rho _{BC}$. Generally, for any tripartite pure state $|\psi\rangle_{A'A''A'''}$,
we  have $S(\rho_{A'|A''})=E(\rho_{A'A'''})$, where $A', A'', A'''$ correspond to any permutations of $A, B, C$. Further more, we find,
\begin{eqnarray}
D^{\leftarrow }(\rho _{A'A''}) &=&\widetilde{I}(\rho _{A'A''})-(S(\rho _{A'})-S(\rho_{A'|A''}))
\notag \\
&=&S(\rho _{A''})-S(\rho _{A'''})+E(\rho _{A'A'''}),
\notag \\
D^{\rightarrow }(\rho _{A'A''}) &=&\widetilde{I}(\rho _{A'A''})-(S(\rho _{A''})-S(\rho_{A''|A'}))
\notag \\
&=&S(\rho _{A'})-S(\rho _{A'''})+E(\rho _{A''A'''}).
\label{LEFt}
\end{eqnarray}
Inserting (\ref{LEFt}) into (\ref{monogam1}),(\ref{monogamyde}), we have that
\begin{eqnarray}
\bigtriangleup_{D_A}^{\leftarrow }&=&S(\rho _{A})-E(\rho _{AC})-E(\rho _{AB}),
\notag \\
\bigtriangleup_{D_A}^{\rightarrow }&=&S(\rho _{B})+S(\rho _{C})-S(\rho _{A})-2E(\rho _{BC}).
\label{monogam2}
\end{eqnarray}
Similarly, $\bigtriangleup_{D_B}^{\leftarrow }, \bigtriangleup_{D_C}^{\leftarrow }$ and $\bigtriangleup_{D_B}^{\rightarrow },
\bigtriangleup_{D_C}^{\rightarrow }$ can be obtained by permutating the indices of (\ref{monogam1}) and (\ref{monogamyde}).
Combining those results, we have (\ref{1deficit}), (\ref{relation}), which completes the proof.

The above relations are interesting. They tell us that the two kinds of monogamy inequalities
which was studied previously \cite{MA,giorgi,XI} actually are not independent. We find that  $\bigtriangleup_{D_A}^{\leftarrow }$,
which is the defined monogamy deficit having a coherent measurement taken on two parties $B$ and $C$, is precisely equal to the
arithmetic mean of $\bigtriangleup_{D_B}^{\rightarrow }$ and $\bigtriangleup_{D_C}^{\rightarrow }$ in which
the measurements are only performed individually on $B$ and $C$. To be explicit, the measurement for left hand side (l.h.s.) of equality
(\ref{1deficit}) is a coherent measurement on ``$BC$'' while on right hand side (r.h.s.),
local measurements on ``$B$'' and ``$C$'' are performed.
In Ref. \cite{giorgi}, a transition from satisfying the monogamy inequality to violation of monogamy inequality
is given, where the positive or negative of $\bigtriangleup_{D_A}^{\leftarrow }$ are studied.
By the definition of monogamy deficit
$\bigtriangleup_{D_A}^{\leftarrow }=D^{\leftarrow }(\rho_{A|BC})-D^{\leftarrow }(\rho_{AB})-D^{\leftarrow }(\rho_{AC})$,
it is apparent that a coherent measurement on ``$BC$'' is necessary. Here our result (\ref{1deficit}) shows that
instead of a coherent measurement, local measurements individually on ``$B$'' and ``$C$'' can be performed to find
this conclusion. We remark that local operation is much easier to be implemented than coherent measurement.
Those results reveal the hidden relationship in monogamy deficit for quantum discord
where the coherent measurement is replaced by local measurements.

In the previous paragraphs, we have already discussed the relationship
between these two monogamy deficit. Now let us consider the difference
between them. The difference between the two monogamy deficit can be expressed
as the following form
\begin{eqnarray}
E\left(\rho_{AB}\right) -\bar{\omega}_{A\mid B}^{+}{}_{\left(D\right)}=\frac{1}{2} ( \triangle_{D_{C}}^{\leftarrow}- \triangle_{D_{C}}^{\rightarrow}).
\label{deiference}
\end{eqnarray}
Where the $\bar{\omega}_{A\mid B}^{+}{}_{\left(D\right)}$ is the
average of $D^{\rightarrow}\left(\rho_{BA}\right)$ and $D^{\rightarrow}\left(\rho_{AB}\right)$.
Here we give a simple proof of the above formula.
By using the Koashi-Winter relation, we have
\begin{eqnarray*}
\triangle_{D_{B}}^{\rightarrow}=I^{\rightarrow}\left(\rho_{BA}\right)
-D^{\rightarrow}\left(\rho_{BA}\right),
\end{eqnarray*}
\begin{eqnarray*}
\triangle_{D_{A}}^{\leftarrow}=I^{\rightarrow}\left(\rho_{BA}\right)
-E\left(\rho_{AB}\right).
\end{eqnarray*}
By subtracting the above equalities, we have
\begin{eqnarray*}
\triangle_{D_{B}}^{\rightarrow} -\triangle_{D_{A}}^{\leftarrow}=E\left(\rho_{AB}\right) - D^{\rightarrow}\left(\rho_{BA}\right),
\end{eqnarray*}
by exchanging the symbols of  $A$ and $B$, we have the similar equation
\begin{eqnarray*}
\triangle_{D_{A}}^{\rightarrow} - \triangle_{D_{B}}^{\leftarrow}=E\left(\rho_{BA}\right)- D^{\rightarrow}\left(\rho_{AB}\right).
\end{eqnarray*}
Combining the above formulas, the $E\left(\rho_{AB}\right)$ - $\bar{\omega}_{A\mid B}^{+}{}_{\left(D\right)}$ can
be expressed as follows
\begin{eqnarray*}
E\left(\rho_{AB}\right) - \bar{\omega}_{A\mid B}^{+}{}_{\left(D\right)}=
\frac{1}{2} ( \triangle_{D_{A}}^{\rightarrow}+ \triangle_{D_{B}}^{\rightarrow})
- \frac{1}{2} (\triangle_{D_{A}}^{\leftarrow} + \triangle_{D_{B}}^{\leftarrow}).
\end{eqnarray*}
Substituting (\ref{1deficit}) and (\ref{relation}) into the above equation, we
have Eq.(\ref{deiference}), which completes the proof.

This formula is meaningful, it shows that the difference between
these two monogamy deficit depends on the balance of entanglement
of formation (EOF) and the average of discord. In other words, if the EOF between
$A$ and $B$ is greater than or equal to the average of discord, the first monogamy
deficit which contains measurements on $A$ and $B$ must be greater than
or equal to the second monogamy deficit which only contains local
measurement on $C$. Especially when $E\left(\rho_{AB}\right)$ is equal to the $\bar{\omega}_{A\mid B}^{+}{}_{\left(D\right)}$,
these two monogamy deficit is equivalent. In other words,
in this case, when we consider the monogamy property of quantum discord,
we only need to know one of them. Since the second monogamy deficit only contains local measurement on one party, it is easier to calculate and to be used in practical applications. Furthermore, the above formula
provides a new physical interpretation of the difference between EOF
and discord in the average sense.

As an application of the relationship between the two monogamy deficit,  for tripartite pure states, we provide an unified physical significance for these two monogamy
deficit. To see this, we first consider the equivalent expression of the second
monogamy deficit. According to the results in  \cite{FF} and \cite{XI},
we have $I^{\rightarrow}\left(\rho_{BA}\right)-D^{\rightarrow}\left(\rho_{BA}\right)=I^{\rightarrow}\left(\rho_{BC}\right)-D^{\rightarrow}\left(\rho_{BC}\right)=\widetilde{I}(\rho _{AC})-2E(\rho _{AC})$
and $\triangle_{D_{B}}^{\rightarrow}=\widetilde{I}(\rho _{AC})-2E(\rho _{AC})$.
Combing the two equations, we have that
\begin{eqnarray}
\triangle_{D_{B}}^{\rightarrow}=I^{\rightarrow}\left(\rho_{BA}\right)-D^{\rightarrow}\left(\rho_{BA}\right)=
I^{\rightarrow}\left(\rho_{BC}\right)-D^{\rightarrow}\left(\rho_{BC}\right).\label{monogam221}
\end{eqnarray}
Where $I^{\rightarrow}\left(\rho_{BA}\right)$ represents the classical
correlation, $D^{\rightarrow}\left(\rho_{BA}\right)$ represents quantum
discord. By exchanging the subscript, we have $\triangle_{D_{C}}^{\rightarrow}=I^{\rightarrow}\left(\rho_{CA}\right)-D^{\rightarrow}\left(\rho_{CA}\right)=I^{\rightarrow}\left(\rho_{CB}\right)-D^{\rightarrow}\left(\rho_{CB}\right)$.
It tells us that this monogamy deficit is equivalent to the difference
between classical correlation and quantum correlation. Since the classical
correlation can be regarded as locally accessible mutual information
(LAMI)\cite{FFF}, while the quantum correlation can be seen as locally inaccessible
mutual information (LIMI). In this sense, this monogamy deficit tells
us that how much mutual information can be extracted from a tripartite
pure state by using local measurement on one party.
To be more explicit, the monogamy inequality holds if and only if
more than half of the mutual information between $AC$ or $BC$ can be accessed  through
local measurement performed on $C$.

The above result provides a interesting relationship between the second monogamy deficit and the difference between LAMI and LIMI for arbitrary tripartite pure states. In the following, we generalize the relationship to arbitrary $N$-partite pure states. For arbitrary $N$-partite pure states $\rho_{A_1\cdots A_N}$, we have
\begin{eqnarray}
\triangle_{D_{A_{1\left(N\right)}}}^{\rightarrow}-
\triangle_{D_{A_{1\left(N-1\right)}}}^{\rightarrow}=I^{\rightarrow}\left(\rho_{A_{1}A_{N}}\right)
- D^{\rightarrow}\left(\rho_{A_{1}A_{N}}\right).
\end{eqnarray}
Where $\triangle_{D_{A_{1\left(N\right)}}}^{\rightarrow}$ represent the second monogamy deficit for arbitrary $N$-partite pure state and is given by
\begin{eqnarray*}
\triangle_{D_{A_{1\left(N\right)}}}^{\rightarrow}=D^{\rightarrow}\left(\rho_{A_{1}\mid A_{2}\cdots A_{N}}\right)
-\sum_{i=2}^{N} D^{\rightarrow}\left(\rho_{A_{1}A_{i}}\right),
\end{eqnarray*}
similarly, the second monogamy deficit $\triangle_{D_{A_{1\left(N-1\right)}}}^{\rightarrow}$ for its $(N-1)$-partite subsystem is
\begin{eqnarray*}
\triangle_{D_{A_{1\left(N-1\right)}}}^{\rightarrow} =D^{\rightarrow}\left(\rho_{A_{1}\mid A_{2}\cdots A_{N-1}}\right)
-\sum_{i=2}^{N-1} D^{\rightarrow}\left(\rho_{A_{1}A_{i}}\right).
\end{eqnarray*}
So we have
\begin{eqnarray*}
\triangle_{D_{A_{1\left(N\right)}}}^{\rightarrow}- \triangle_{D_{A_{1\left(N-1\right)}}}^{\rightarrow}&=&D^{\rightarrow}\left(\rho_{A_{1}\mid A_{2}\cdots A_{N}}\right)-D^{\rightarrow}\left(\rho_{A_{1}\mid A_{2}\cdots A_{N-1}}\right).
\notag\\
& &-
D^{\rightarrow}\left(\rho_{A_{1}A_{N}}\right).
\end{eqnarray*}
Using the Koashi-Winter relationship and considering the property of pure states, it is easy to show
\begin{eqnarray*}
D^{\rightarrow}\left(\rho_{A_{1}A_{N}}\right)
+ I^{\rightarrow}\left(\rho_{A_{1}\mid A_{2}\cdots A_{N-1}}\right)
= S\left(\rho_{A_{1}}\right).
\end{eqnarray*}
Since $S\left(\rho_{A_{1}}\right)=D^{\rightarrow}\left(\rho_{A_{1}\mid A_{2}\cdots A_{N}}\right)$ for pure state, the above results can be rewritten as follows
\begin{eqnarray*}
\triangle_{D_{A_{1\left(N\right)}}}^{\rightarrow}- \triangle_{D_{A_{1\left(N-1\right)}}}^{\rightarrow}
&=&I^{\rightarrow}\left(\rho_{A_{1}\mid A_{2}\cdots A_{N-1}}\right)
- D^{\rightarrow}\left(\rho_{A_{1}\mid A_{2}\cdots A_{N-1}}\right)\\
&=&I^{\rightarrow}\left(\rho_{A_{1}A_{N}}\right) - D^{\rightarrow}\left(\rho_{A_{1}A_{N}}\right).
\end{eqnarray*}
Which completes the proof.

This equation tells us that the difference between the second monogamy deficit
of $N$-partite system and its $(N-1)$-partite subsystem is equivalent
to the difference between classical correlation and quantum discord.
That is to say, the difference between
LAMI and  LIMI can tell
us that which system is more monogamous, the $N$-partite system or its
$(N-1)$-partite subsystem. In other words, if we can extract more than
half of the mutual information between $A_{1}$ and $A_{N}$ through
local measurements on $A_{1}$, the $N$-partite system must be more
monogamous than its $(N-1)$-partite subsystem. As we all know, in studying entanglement of a tripartite system, the monogamy deficit of entanglement can be seen as a tripartite correlation which is called tangle, or genuine entanglement. Similarly, for discord of an $N$-partite system, the monogamy deficit can also be seen as a type of multipartite correlation which beyond the usual bipartite correlations. In this sense, the $N$-partite
system contains more multipartite correlation than its $(N-1)$-partite
subsystem if and only if we can acquire at least half of the mutual
information between $A_{1}$ and $A_{N}$ through local measurements
on $A_{1}$. When $N=3$, the above result goes back to formula (\ref{monogam221}).

Now we can give a similar equivalent expression of $\triangle_{D_{C}}^{\leftarrow}$. In order to achieve this purpose,
we first define the average of classical correlation and discord.
As we all know, the discord and classical correlation are asymmetry
quantities.  By using the asymmetry of $I^{\rightarrow}\left(\rho_{BA}\right)$ and $I^{\rightarrow}\left(\rho_{AB}\right)$,
we  define  the average of classical correlation $\bar{\omega}_{A\mid B}^{+}{}_{\left(I\right)}$
% $I^{\rightarrow}\left(\rho_{BA}\right)$
%and $I^{\rightarrow}\left(\rho_{AB}\right)$
and the  average discord $\bar{\omega}_{A\mid B}^{+}{}_{\left(D\right)}$ \cite{FFF},
\begin{eqnarray*}
\bar{\omega}_{A\mid B}^{+}{}_{\left(I\right)}=\frac{1}{2}(I^{\rightarrow}\left(\rho_{BA}\right)+I^{\rightarrow}\left(\rho_{AB}\right)),\nonumber\\
\bar{\omega}_{A\mid B}^{+}{}_{\left(D\right)}=\frac{1}{2}(D^{\rightarrow}\left(\rho_{BA}\right)+D^{\rightarrow}\left(\rho_{AB}\right)).
\end{eqnarray*}
From this definition, combing Eq. (\ref{1deficit}) and (\ref{monogam221}) , we have
\begin{eqnarray}
\triangle_{D_{C}}^{\leftarrow}&=&\frac{1}{2}(\triangle_{D_{A}}^{\rightarrow}+\triangle_{D_{B}}^{\rightarrow})\nonumber\\
&=&\frac{1}{2}(I^{\rightarrow}\left(\rho_{AB}\right)-D^{\rightarrow}\left(\rho_{AB}\right)+I^{\rightarrow}\left(\rho_{BA}\right)-D^{\rightarrow}\left(\rho_{BA}\right))\nonumber\\
&=&\bar{\omega}_{A\mid B}^{+}{}_{\left(I\right)}-\bar{\omega}_{A\mid B}^{+}{}_{\left(D\right)}.\label{monogam22121}
\end{eqnarray}
This formula means that the monogamy deficit
which needs a coherent measurement performed on two parties $A$ and
$B$ is equivalent to $\bar{\omega}_{A\mid B}^{+}{}_{\left(I\right)}-\bar{\omega}_{A\mid B}^{+}{}_{\left(D\right)}$,
which is the difference between the average of classical correlations
and quantum correlations where local measurements are made individually on $A$ and
$B$. In other words, according to previous view, this monogamy deficit
represents our ability to extract the mutual information by performing
local measurements on these two parties. According to the above definition,
simply we have the relation,
$\bar{\omega}_{A\mid B}^{+}{}_{\left(I\right)}+\bar{\omega}_{A\mid B}^{+}{}_{\left(D\right)}=\widetilde{I}\left(\rho_{AB}\right)$.
Which means that  the monogamy inequality holds if and only if we can acquire at least half of the mutual information   through
local measurements in the average sense.

The Eq. (\ref{monogam22121}) presented above can also be used as a criterion to check whether a given
tripartite pure state belongs to GHZ class or W class state under stochastic local operations
and classical communication (SLOCC). According to the results in Ref. \cite{MA},
we have that a tripartite pure state belongs to GHZ class state  if and only if $\triangle_{D_{C}}^{\leftarrow}\geq 0$,
otherwise it belongs to W class state. % if and only if $\triangle_{D_{C}}^{\leftarrow}<0$.
In other words, we can say that a tripartite pure state belongs to GHZ class state if and only if
the LAMI is always greater than or equal to the LIMI in the average
sense when local measurements are performed on $A$ and $B$. While a tripartite pure state belongs to W class state if and only if
LAMI is less than LIMI in the average sense when local measurements are performed on $A$ and $B$. In addition, a tripartite pure state
belongs to GHZ class or W class depends on whether one can acquire no less than half of the mutual information   through
local measurements in the average sense.

The above result tells us that there is a interesting relationship between the first monogamy deficit and the difference between the average of LAMI and LIMI for arbitrary tripartite pure states. In fact, we can generalize the relationship to arbitrary $N$-partite pure states. For arbitrary $N$-partite pure states $\rho_{A_1\cdots A_N}$, we have

\begin{eqnarray}
\triangle_{D_{A_{1\left(N\right)}}}^{\leftarrow}-\triangle_{D_{A_{1\left(N-1\right)}}}^{\leftarrow}
= \bar{\omega}_{\left(A_{2\cdots}A_{N-1}\right)\mid A_{N}}^{+}{}_{\left(I\right)}
- \bar{\omega}_{\left(A_{2\cdots}A_{N-1}\right)\mid A_{N}}^{+}{}_{\left(D\right)}.
\end{eqnarray}
Where $\triangle_{D_{A_{1\left(N\right)}}}^{\leftarrow}$ represent the first monogamy deficit for arbitrary $N$-partite pure state and is given by
\begin{eqnarray*}
\triangle_{D_{A_{1\left(N\right)}}}^{\leftarrow}=D^{\leftarrow}\left(\rho_{A_{1}\mid A_{2}\cdots A_{N}}\right)
- \sum_{i=2}^{N}D^{\leftarrow}\left(\rho_{A_{1}A_{i}}\right),
\end{eqnarray*}
similarly, the first monogamy deficit $\triangle_{D_{A_{1\left(N-1\right)}}}^{\leftarrow}$ for its $(N-1)$-partite subsystem is
\begin{eqnarray*}
\triangle_{D_{A_{1\left(N-1\right)}}}^{\leftarrow}=D^{\leftarrow}\left(\rho_{A_{1}\mid A_{2}\cdots A_{N-1}}\right)
- \sum_{i=2}^{N-1}D^{\leftarrow}\left(\rho_{A_{1}A_{i}}\right).
\end{eqnarray*}
So we have
\begin{eqnarray*}
\triangle_{D_{A_{1\left(N\right)}}}^{\leftarrow}-\triangle_{D_{A_{1\left(N-1\right)}}}^{\leftarrow}
&=&D^{\leftarrow}\left(\rho_{A_{1}\mid A_{2}\cdots A_{N}}\right)
-
D^{\leftarrow}\left(\rho_{A_{1}\mid A_{2}\cdots A_{N-1}}\right)
\notag\\
& &-
D^{\leftarrow}\left(\rho_{A_{1}A_{N}}\right).
\end{eqnarray*}
Let's regard the $N$-partite pure state $\rho_{A_{1}\cdots A_{N}}$ as the tripartite pure state $\rho_{A_{1}\left(A_{2}\cdots A_{N-1}\right)A_{N}}$.
In this sense, the right-hand side of the above formula can be considered
to be an first monogamy deficit for this equivalent tripartite pure
state. By using the previous equation (\ref{monogam22121}), the above
formula can be expressed as follows
\begin{eqnarray*}
\triangle_{D_{A_{1\left(N\right)}}}^{\leftarrow}-\triangle_{D_{A_{1\left(N-1\right)}}}^{\leftarrow}
= \bar{\omega}_{\left(A_{2\cdots}A_{N-1}\right)\mid A_{N}}^{+}{}_{\left(I\right)}
- \bar{\omega}_{\left(A_{2\cdots}A_{N-1}\right)\mid A_{N}}^{+}{}_{\left(D\right)}.
\end{eqnarray*}
Which completes the proof.

This equation provides that the difference between the monogamy deficit of $N$-partite
system and its $(N-1)$-partite subsystem is equivalent to the difference
between the average of classical correlations and quantum correlations.
That is to say, the right hand side of this equation can tell us that which system is more monogamous, the $N$-partite
system or its $(N-1)$-partite subsystem. Then we can say
that the $N$-partite system must be more monogamous than its $(N-1)$-partite
subsystem if and only if the LAMI is always greater than or equal
to the LIMI in the average sense.  Additionally, similar as the explanation of the formula (10), the $N$-partite
system contains more multipartite correlation than its $(N-1)$-partite
subsystem if and only if we can acquire at least half of the mutual
information through measurements performed on $\left(A_{2\cdots}A_{N-1}\right)$
and $A_{N}$ in the average sense. When $N=3$, the above result returns
to  formula  (\ref{monogam22121}).

As a short summary for previous discussion, for tripartite pure state, we demonstrate a relation which connects two types of monogamy inequalities
for quantum discord and provide the difference between them. By using this relation, we get an unified view for these two monogamy inequalities. That is, for arbitrary tripartite pure states, both of the monogamy
inequalities hold only if one can extract more than half of mutual information by using local measurement.

\subsection{Monogamy deficit for discord and other quantum correlations}\label{squash}

The squashed entanglement is an entanglement monotone for bipartite quantum
states introduced by Christandl and Winter \cite{FE}. For bipartite state $\rho_{AB}$,
the squashed entanglement is given by
\begin{eqnarray*}
E_{sq}(\rho_{AB})\equiv \frac{1}{2}inf\widetilde{I}(\rho_{A:B\mid C}),
\end{eqnarray*}
here $\widetilde{I}(\rho_{A:B\mid C})$ is the conditional mutual information of  $\rho_{ABC}$ with respect to particle $C$ (see Eq. (\ref{conditionm})), $\rho_{ABC}$ is the extension of $\rho_{AB}$  and the infimum is taken over all extensions of $\rho_{AB}$ such that $\rho _{AB}={\rm Tr}_C\rho _{ABC}$.
The squashed entanglement has many important properties \cite{KO}. For example,
\emph{(1) The squashed entanglement is upper bounded by entanglement of formation.
(2) For any tripartite state $\rho_{ABC}$, we have $E_{sq}(\rho_{AB})+E_{sq}(\rho_{AC})\leq E_{sq}(\rho_{A(BC)})$.}

By using the property (2), we can generalize the concept of monogamy deficit for squashed entanglement.
We define the monogamy deficit for squashed entanglement  as
\begin{eqnarray}
\triangle_{(E_{sq(C)})}=E_{sq}(\rho_{C(AB)})-E_{sq}(\rho_{CA})-E_{sq}(\rho_{CB}).
\label{squuud}
\end{eqnarray}
Similarly, one can also define the monogamy deficit $\triangle_{E_C}=E(\rho_{C(AB)})-E(\rho_{CA})-E(\rho_{CB})$  for the entanglement of formation.

For the monogamy deficit for squashed entanglement, we can have the following result.
For any tripartite pure state $\rho_{ABC}$, we may observe,
 \begin{eqnarray}
\triangle_{(E_{sq(C)})}\geq\max \{   \triangle_{D_{C}}^{\leftarrow},0\}=\max\left\{ \bar{\omega}_{A\mid B}^{+}{}_{\left(I\right)}-\bar{\omega}_{A\mid B}^{+}{}_{\left(D\right)},0\right\}.
\label{relati}
\end{eqnarray}
The proof of this relation is presented below.
By the above property (2) of monogamy deficit for squashed entanglement, we have $\triangle_{(E_{sq(C)})}\geq 0$, and the second equality
is given by Eq.(\ref{monogam22121}). We now only need to prove $\triangle_{(E_{sq(C)})}\geq \triangle_{E_{C}}$  and $\triangle_{E_{C}}=\triangle_{D_{C}}^{\leftarrow}$. For tripartite pure state, $E_{sq}(\rho_{C(AB)})=E(\rho_{C(AB)})$, thus we have
\begin{eqnarray*}
\triangle_{(E_{sq(C)})}&=& E_{sq}(\rho_{C(AB)})-E_{sq}(\rho_{CA})-E_{sq}(\rho_{CB})\nonumber\\
&=& E(\rho_{C(AB)})-E_{sq}(\rho_{CA})-E_{sq}(\rho_{CB})\nonumber\\
&\geq & E(\rho_{C(AB)})-E(\rho_{CA})-E(\rho_{CB}).\nonumber\\
&=& \triangle_{E_{C}}.
\end{eqnarray*}
For tripartite pure state,  we have $E(\rho_{C(AB)})=D^{\leftarrow }(\rho_{C|AB})=S(\rho_{C})$,
combing $E(\rho_{CA})+E(\rho_{CB})=D^{\leftarrow}\left(\rho_{CA}\right)+D^{\leftarrow}\left(\rho_{CB}\right)$ in
Ref. \cite{FF1}, thus we have $\triangle_{E_{C}}=\triangle_{D_{C}}^{\leftarrow}$. That is
$\triangle_{(E_{sq(C)})}\geq \triangle_{E_{C}}=\triangle_{D_{C}}^{\leftarrow}$.
Now we know that Eq. (\ref{relati}) is true.

The quantum work deficit is an important information-
theoretic measure of quantum correlation introduced by Oppenheim \emph{et al.} \cite{JO}.
For an arbitrary bipartite state $\rho_{AB}$, the
quantum work-deficit is defined as
\begin{eqnarray}
\Delta(\rho_{AB})=I_G(\rho_{AB})-I_L(\rho_{AB}),
\label{workdefi}
\end{eqnarray}
where $I_G(\rho_{AB})$ represents the thermodynamic ``work'' that can be extracted from $\rho_{AB}$
by ``closed global operations'', $I_L(\rho_{AB})$ represents the thermodynamic ``work'' that can be extracted from $\rho_{AB}$
by closed local operation and classical communication (CLOCC) \cite{key-12}.
Further more, the one side work deficit $\Delta^\rightarrow(\rho_{AB})$ ($\Delta^\leftarrow(\rho_{AB})$)
means that CLOCC is restricted on projection measurements at one particle $A$ ($B$). According to Ref.\cite{key-12}, the one side
work deficit is lower bounded by quantum discord, that is
\begin{eqnarray*}
D^\rightarrow(\rho_{AB})\leq \Delta^\rightarrow(\rho_{AB}),\nonumber\\
D^\leftarrow(\rho_{AB}) \leq \Delta^\leftarrow(\rho_{AB}).
\end{eqnarray*}
Similar to quantum discord and squashed entanglement, we provide the definition of the monogamy deficit for work deficit.
The two kinds of monogamy deficit for work deficit are given as,
\begin{eqnarray}
\bigtriangleup_{\Delta_A}^{\leftarrow }=\Delta^{\leftarrow }(\rho_{A:BC})-\Delta^{\leftarrow }(\rho_{AB})-\Delta^{\leftarrow }(\rho_{AC}),
\label{monogam11}
\end{eqnarray}
\begin{eqnarray}
\bigtriangleup_{\Delta_A}^{\rightarrow }=\Delta^{\rightarrow }(\rho_{A:BC})-\Delta^{\rightarrow }(\rho_{AB})-\Delta^{\rightarrow }(\rho_{AC}).
\label{monogamyde2}
\end{eqnarray}
The first  definition involves a POVM coherently performed on $B$ and $C$ together,
and the other involve measurements
only on $A$.

For the monogamy deficit for work deficit, we present the following observations.
For any tripartite pure state $\rho_{ABC}$, we have that
 \begin{eqnarray}
& &\bigtriangleup_{\Delta_A}^{\leftarrow }\leq \bigtriangleup_{D_A}^{\leftarrow }\leq \triangle_{(E_{sq(A)})},
\label{relati12}\\
& & \bigtriangleup_{\Delta_A}^{\rightarrow }\leq \bigtriangleup_{D_A}^{\rightarrow }.
\label{relati123}
\end{eqnarray}
The correctness of these observations are presented below.
The inequality $\bigtriangleup_{D_A}^{\leftarrow }\leq \triangle_{(E_{sq(A)})}$ is from (\ref{relati}). For tripartite pure state $\rho_{ABC}$, we have $\Delta^{\leftarrow }(\rho_{A:BC})=D^{\leftarrow }(\rho_{A|BC})=S(\rho_A)$,
$D^{\leftarrow }(\rho_{AB})\leq \Delta^{\leftarrow }(\rho_{AB})$, $D^{\leftarrow }(\rho_{AC})\leq \Delta^{\leftarrow }(\rho_{AC})$\cite{key-12}.
Which implies that $\bigtriangleup_{\Delta_A}^{\leftarrow }\leq \bigtriangleup_{D_A}^{\leftarrow }$. Similarly, we can prove (\ref{relati123}).

The physical interpretation of Eq. (\ref{relati12}) can be like the following.
The monogamy property for work deficit implies the monogamy property for quantum discord, entanglement of formation and squashed entanglement for any pure state. In this case, we have $\bigtriangleup_{\Delta_A}^{\leftarrow }\leq  \bar{\omega}_{B\mid C}^{+}{}_{\left(I\right)}-\bar{\omega}_{B\mid C}^{+}{}_{\left(D\right)}  \leq \triangle_{(E_{sq(A)})} $. In addition, we can extract more than half of the mutual
information between $BC$ in the average sense through local measurements.
Combing (\ref{monogam221}) and (\ref{relati123}), we have $\bigtriangleup_{\Delta_A}^{\rightarrow }\geq 0$ which implies one can extract at least half of
the mutual information between $AB$ or $AC$ by using local measurement of $A$.

\section{Necessary and sufficient criteria for non-negative monogamy deficit $\bigtriangleup_{D_A}^\rightarrow$}\label{two}
As is shown in \cite{MA}, a  necessary and sufficient condition for discord  to be monogamous is
$\bigtriangleup_{D_A}^{\leftarrow }\ge 0$, see Eq. (\ref{monogam1}) for definition.
%It can be shown that monogamy deficit $\bigtriangleup_{D_A}^\rightarrow$ in Eq. (\ref{monogam1})
%equals to the difference of unmeasured interaction information $\widetilde{I}(\rho_{ABC})$ and $I(\rho_{ABC})$
%(the interaction information after measurement).
Since we present two kinds of monogamy deficit,
an interesting question is what does it means if the measurement
is taken on another side,
%that does there exist
%similar necessary and sufficient condition for
 $\bigtriangleup_{D_A}^\rightarrow \ge 0$? In this section,
we will consider this question and prove a similar necessary and sufficient condition for the second kinds of monogamy inequality.
%find the answer is true.
 We first present some definitions about mutual
information, conditional mutual information with respect to a single particle $A$.

For a tripartite state $\rho_{ABC}$, the unmeasured conditional mutual information with respect to particle $A$
is given as,
\begin{eqnarray}
\widetilde{I}_A(\rho_{B:C|A})\equiv \widetilde{S}(\rho_{B|A})+\widetilde{S}(\rho_{C|A})-\widetilde{S}(\rho_{BC|A}),
\label{conditionm}
\end{eqnarray}
and the interrogated conditional mutual information with respect to particle $A$ is,
\begin{eqnarray}
I_A(\rho_{B:C|A})\equiv S(\rho_{B|A})+S(\rho_{C|A})-S(\rho_{BC|A}).
\end{eqnarray}
By the strong subadditivity of von Neumann entropy,
we know, $\widetilde{I}_A(\rho_{B:C|A})\ge 0$, $I_A(\rho_{B:C|A})\ge 0$, both are non-negative.

We next propose the concept of interaction information. The (unmeasured)
interaction information $\widetilde{I}_A(\rho_{ABC})=\widetilde{I}_A(\rho_{B:C|A})-\widetilde{I}(\rho_{BC})$.
By simple calculating, one may observe that $\widetilde{I}_A(\rho_{ABC})=S(\rho_{AB})+S(\rho_{AC})+S(\rho_{BC})
-(S(\rho_{A})+S(\rho_{B})+S(\rho_{C}))-S(\rho_{ABC})= \widetilde{I}(\rho_{ABC})$,
this is the  interaction information
defined in Ref. \cite{MA}.

 For the state $\rho_{ABC}$ and a given measurement $\{\mathcal{M}_i^A\}$, an
interrogated interaction information with respect to $A$ is given as,
\begin{eqnarray}
I_A(\rho_{ABC})_{\{\mathcal{M}_i^A\}}\equiv I_A(\rho_{B:C|A})-I_A(\rho_{BC})_{\{\mathcal{M}_i^A\}}.
\label{interaction-inf}
\end{eqnarray}
Since $I_A(\rho_{BC})_{\{\mathcal{M}_i^A\}}$ do not have particle $A$, we have $I_A(\rho_{BC})_{\{\mathcal{M}_i^A\}}=\widetilde{I}(\rho_{BC})$,
which does not involve any measurement.
Given a tripartite quantum state $\rho_{ABC}$, $I_A(\rho_{ABC})_{\{\mathcal{M}_i^A\}}$ represents the interaction information with respect to $A$, which
is defined in (\ref{interaction-inf}). For this definition (\ref{interaction-inf}), the first term of the right hand side
is the conditional mutual information of $B$, $C$ when $A$ is present and measured,
the second term is the mutual information of $BC$ where $A$ is absent.
Here, $I_A(\rho_{ABC})_{\{\mathcal{M}_i^A\}}$ measures the effect on the amount of correlation shared between $B$ and $C$ by measuring $A$.
A positive interaction information with respect to $A$ means the presentation of $A$ can enhance the total correlation
between $B$ and $C$, while negative interaction information with respect to $A$ means the presentation of $A$ inhibits the total correlation
between $B$ and $C$. $I_A(\rho_{ABC})_{\{\mathcal{M}_i^A\}}$ has the similar
property as $I(\rho_{ABC})$ proposed in Ref. \cite{MA} and can be read as a necessary and sufficient criteria for
a monogamy inequality. We next have the following theorem.

\begin{thm}
For any state $\rho_{ABC}$, $D^\rightarrow(\rho_{AB})+D^\rightarrow(\rho_{AC})\leq D^\rightarrow(\rho_{A|BC})$ if and only if
the interrogated interaction information with respect to $A$ is less than or equal to the unmeasured interaction information
with respect to $A$.

\end{thm}
\begin{proof}
We only need to calculate the monogamy deficit $\bigtriangleup_{D_A}^\rightarrow$,
\begin{eqnarray}
\bigtriangleup_{D_A}^{\rightarrow } &=&D^{\rightarrow }(\rho _{A|BC})-D^{\rightarrow }(\rho _{AB})-D^{\rightarrow }(\rho _{AC})
\notag \\
&=&S(\rho _{BC|A})-\widetilde{S}(\rho _{BC|A})-(S(\rho _{B|A})-\widetilde{S}(\rho _{B|A}))
\notag \\
& &-(S(\rho _{C|A})-\widetilde{S}(\rho _{C|A}))
\notag \\
&=&S(\rho _{BC|A})-S(\rho _{B|A})-S(\rho _{C|A})+I(\rho _{BC})
\notag \\
& &-(\widetilde{S}(\rho _{BC|A})-\widetilde{S}(\rho _{B|A})-\widetilde{S}(\rho _{C|A})+\widetilde{I}(\rho _{BC}))
\notag \\
&=&\widetilde{I}_A(\rho_{ABC})-I_A(\rho_{ABC})_{\{\mathcal{M}_i^A\}}.
\label{disefi}
\end{eqnarray}
From (\ref{disefi}), we have $\bigtriangleup_{D_A}^{\rightarrow }\geq 0$ if and only if $\widetilde{I}_A(\rho_{ABC})\geq I_A(\rho_{ABC})_{\{\mathcal{M}_i^A\}}$
which completes the proof.
\end{proof}

For pure state, we have $\widetilde{I}_A(\rho_{ABC})=0$, and the monogamy deficit of quantum discord is equivalent to the non-positivity of the interrogated
information with respect to $A$.

To see a transition from violation to observation of monogamy, we
consider a family of states \cite{giorgi},
\begin{eqnarray}
|\widetilde{\psi}\left(p,\varepsilon\right)\rangle&=&\sqrt{p\varepsilon}|000\rangle+\sqrt{p(1-\varepsilon)}|111\rangle
\nonumber \\
&&+\sqrt{\frac{1-p}{2}}(|101\rangle+|110\rangle).
\end{eqnarray}
Note that $|\widetilde{\psi}(\frac{1}{3},1)\rangle$ is the
 maximally
entangled W state $\sqrt{\frac{1}{3}}(|000\rangle+|101\rangle+|110\rangle)$,
while $|\widetilde{\psi}(1,\frac{1}{2})\rangle$ is the GHZ state, $\sqrt{\frac{1}{2}}(|000\rangle+|111\rangle)$.
In Fig.1, $\triangle_{D_{A}}^{\rightarrow}=D^{\rightarrow}(\rho_{A|BC})-D^{\rightarrow}(\rho_{AB})-D^{\rightarrow}(\rho_{AC})$
is plotted as a function of $p$ for different values of $\varepsilon$.
From this figure, we can show that the interrogated
information with respect to $A$ can be positive or negative for tripartite pure state. The $I_A(\rho_{ABC})_{\{\mathcal{M}_i^A\}}$ is increasing with the increasing of $\varepsilon$. All of the three lines are very close to each other when $I_A(\rho_{ABC})_{\{\mathcal{M}_i^A\}}$ is positive and the critical point from positive to negative is almost identical for them.	
Especially for the W state, when p approaches to 1, the $I_A(\rho_{ABC})_{\{\mathcal{M}_i^A\}}$ approaches to zero.

%%%%%%%%%%%%%%%%%%%%%%%%%%%%%%%%%%%%%%%%%%%%%%%%%%%%%%%%%%%%%%%%%%%%%%%%%%
\begin{figure}\centerline{
\includegraphics[width=8cm]{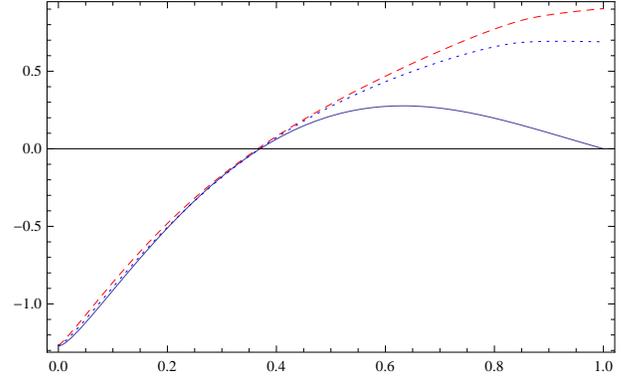}}
\caption{ (Color online) the monogamy deficit $\triangle_{D_{A}}^{\rightarrow}$ for
$|\widetilde{\psi}(p,\varepsilon)\rangle$, $\triangle_{D_{A}}^{\rightarrow}$ as
a function of $p$ for different values of $\varepsilon$ (see the main
text). States are monogamy when the respective curves are positive.
Red dashed line is for $\varepsilon=0.5$, blue dotted line is for $\epsilon=0.75$,
blue solid line is for $\varepsilon=1$. The monogamy deficit is decreasing with the
increasing of $\varepsilon$. All of the three lines are very close to each
other when $\triangle_{D_{A}}^{\rightarrow}$ is polygamy and the critical
point from polygamy to monogamy is almost identical for them. When $\varepsilon\rightarrow1$, $|\widetilde{\psi}(p,\varepsilon)\rangle$
 approach to W states, in this case, when $p\rightarrow1$, $\triangle_{D_{A}}^{\rightarrow}\rightarrow0$. The blue solid line is increased
first and then decreased.
}\label{lct}
\end{figure}
%%%%%%%%%%%%%%%%%%%%%%%%%%%%%%%%%%%%%%%%%%%%%%%%%%%%%%%%%%%%%%%%%%%%%%%

As an application of the necessary and sufficient conditions, we find
an interesting equivalent expression of entanglement of formation for tripartite pure states.

For tripartite pure states,
it is shown that \cite{XI}, $\triangle_{D_{A}}^{\rightarrow}=\widetilde I(\rho_{BC})-2E(\rho_{BC})$.
From the previous discussion, we have the formula,
$\triangle_{D_{A}}^{\rightarrow}=\widetilde{I_{A}}\left(\rho_{B:C\mid A}\right)-I_{A}\left(\rho_{B:C\mid A}\right)$,
which holds for general tripartite mixed states. When we consider the case
of pure states, the above two expressions should be equal. That is
to say, in this case, $\widetilde I(\rho_{BC})-2E(\rho_{BC})=\widetilde{I_{A}}\left(\rho_{B:C\mid A}\right)-I_{A}\left(\rho_{B:C\mid A}\right)$.
At the same time, it is easy to show that $\widetilde{I_{A}}\left(\rho_{B:C\mid A}\right)=S\left(\rho_{BA}\right)+S\left(\rho_{CA}\right)-S\left(\rho_{BCA}\right)-S\left(\rho_{A}\right)=S\left(\rho_{C}\right)+S\left(\rho_{B}\right)-S\left(\rho_{BC}\right)=\widetilde I(\rho_{BC})$.
So we have $E(\rho_{BC})=\frac{1}{2}I_{A}\left(\rho_{B:C\mid A}\right)$.
Thus the interrogated conditional mutual information with respect to $A$
is twice of the entanglement of formation for state $\rho _{BC}$.

\section{Necessary and sufficient criteria for non-negative monogamy deficit for multipartite system}\label{necess}

In this section, we generalize our result to multipartite system.
 We give a necessary and sufficient condition for $\triangle_{D_{A_{1\left(N\right)}}}^{\rightarrow}\ge 0$, where the monogamy
 deficit is defined for multipartite state, $\triangle_{D_{A_{1\left(N\right)}}}^{\rightarrow}=D^{\rightarrow}\left(\rho_{A_{1}\mid A_{2}\cdots A_{N}}\right)-D^{\rightarrow}\left(\rho_{A_{1}A_{2}}\right)-D^{\rightarrow}\left(\rho_{A_{1}A_{3}}\right)-\cdots-D^{\rightarrow}\left(\rho_{A_{1}A_{N}}\right)$.
In order to consider this question, similar as the tripartite state, we next present some definitions
about mutual information, conditional mutual information with respect to
a single particle $A_{1}$.

For a $N$-partite  state $\rho_{A_{1}\cdots A_{N}}$,   the unmeasured conditional
mutual information with respect to particle $A_{1}$ is given as $\widetilde{I_{A_{1}}}\left(\rho_{A_{K}:\left(A_{K+1}\cdots A_{N}\right)\mid A_{1}}\right)=\widetilde{S}\left(\rho_{A_{K}\mid A_{1}}\right)+\widetilde{S}\left(\rho_{\left(A_{K+1}\cdots A_{N}\right)\mid A_{1}}\right)-\widetilde{S}\left(\rho_{A_{K}\left(A_{K+1}\cdots A_{N}\right)\mid A_{1}}\right)$, where $K=2,\ldots, N-1$.
The interrogated conditional mutual information with respect to particle
$A_{1}$ is $I_{A_{1}}\left(\rho_{A_{K}:\left(A_{K+1}\cdots A_{N}\right)\mid A_{1}}\right)=S\left(\rho_{A_{K}\mid A_{1}}\right)+S\left(\rho_{\left(A_{K+1}\cdots A_{N}\right)\mid A_{1}}\right)-S\left(\rho_{A_{K}\left(A_{K+1}\cdots A_{N}\right)\mid A_{1}}\right)$.
By the strong subadditivity of von Neumann entropy, we have $\widetilde{I_{A_{1}}}\left(\rho_{A_{K}:\left(A_{K+1}\cdots A_{N}\right)\mid A_{1}}\right)$
and $I_{A_{1}}\left(\rho_{A_{K}:\left(A_{K+1}\cdots A_{N}\right)\mid A_{1}}\right)$
are both non-negative.

We define the concept of interaction information with respect to $A_{1}$.
The (unmeasured) interaction information is defined as, $\widetilde{I_{A_{1}}}\left(\rho_{A_{1}A_{K}\left(A_{K+1}\cdots A_{N}\right)}\right)=\widetilde{I_{A_{1}}}\left(\rho_{A_{K}:\left(A_{K+1}\cdots A_{N}\right)\mid A_{1}}\right)-\widetilde{I}\left(\rho_{A_{K}\left(A_{K+1}\cdots A_{N}\right)}\right)$.
For the state $\rho_{A_{1}A_{K}\left(A_{K+1}\cdots A_{N}\right)}$ and
a given measurement $\left\{ M_{i}^{A_{1}}\right\} $, an interrogated
interaction information with respect to $A_{1}$ is given as $I_{A_{1}}\left(\rho_{A_{1}A_{K}\left(A_{K+1}\cdots A_{N}\right)}\right)_{\left\{ M_{i}^{A_{1}}\right\} }=I_{A_{1}}\left(\rho_{A_{K}:\left(A_{K+1}\cdots A_{N}\right)\mid A_{1}}\right)-I_{A_{1}}\left(\rho_{A_{K}\left(A_{K+1}\cdots A_{N}\right)}\right)_{\left\{ M_{i}^{A_{1}}\right\} }$,
where the suffix $I_{A_{1}}\left(\rho_{A_{1}A_{K}\left(A_{K+1}\cdots A_{N}\right)}\right)_{\left\{ M_{i}^{A_{1}}\right\} }$ is
used to indicate the measurements on $A_{1}$.

Similar as the above calculating, we find that $\widetilde{I_{A_{1}}}\left(\rho_{A_{1}A_{K}\left(A_{K+1}\cdots A_{N}\right)}\right)=S\left(\rho_{A_{K}\left(A_{K+1}\cdots A_{N}\right)}\right)+S\left(\rho_{A_{K}A_{1}}\right)+S\left(\rho_{\left(A_{K+1}\cdots A_{N}\right)A_{1}}\right)-(S\left(\rho_{A_{K}}\right)+S\left(\rho_{\left(A_{K+1}\cdots A_{N}\right)}\right)+S\left(\rho_{A_{1}}\right))-S\left(\rho_{A_{K}\left(A_{K+1}\cdots A_{N}\right)A_{1}}\right)=\widetilde{I}\left(\rho_{A_{1}A_{K}\left(A_{K+1}\cdots A_{N}\right)}\right)$,
which is the interaction information we have defined. Since $I_{A_{1}}\left(\rho_{A_{K}\left(A_{K+1}\cdots A_{N}\right)}\right)_{\left\{ M_{i}^{A_{1}}\right\} }$ do
not have particle $A_{1}$, we have $I_{A_{1}}\left(\rho_{A_{K}\left(A_{K+1}\cdots A_{N}\right)}\right)_{\left\{ M_{i}^{A_{1}}\right\} }=\widetilde I\left(\rho_{A_{K}\left(A_{K+1}\cdots A_{N}\right)}\right)$,
which also does not involve any measurement as in tripartite case.

Now we can give the necessary and sufficient condition for $\triangle_{D_{A_{1\left(N\right)}}}^{\rightarrow}$ is
no less than zero. We have the following theorem.

\begin{thm}

For any $\rho_{A_{1}\cdots A_{N}}, \triangle_{D_{A_{1\left(N\right)}}}^{\rightarrow}\geq0$ if
and only if the interrogated interaction information with respect to
$A_{1}$ being less than or equal to the unmeasured interaction information
with respect to $A_{1}$.
\end{thm}
\begin{proof}
Here, we only need to calculate the monogamy deficit,
\begin{eqnarray}
\triangle_{D_{A_{1\left(N\right)}}}^{\rightarrow}&=&
D^{\rightarrow}\left(\rho_{A_{1}\mid A_{2}\cdots A_{N}}\right)
-D^{\rightarrow}\left(\rho_{A_{1}A_{2}}\right)-D^{\rightarrow}\left(\rho_{A_{1}A_{3}}\right)\notag \\
& &-\cdots-D^{\rightarrow}\left(\rho_{A_{1}A_{N}}\right)\notag \\
&=&[\widetilde{I_{A_{1}}}\left(\rho_{A_{2}:\left(A_{3}\cdots A_{N}\right)\mid A_{1}}\right)-I_{A_{1}}\left(\rho_{A_{2}:\left(A_{3}\cdots A_{N}\right)\mid A_{1}}\right)]\notag \\
& &
+[\widetilde{I_{A_{1}}}\left(\rho_{A_{3}:\left(A_{4}\cdots A_{N}\right)\mid A_{1}}\right)-I_{A_{1}}\left(\rho_{A_{3}:\left(A_{4}\cdots A_{N}\right)\mid A_{1}}\right)]\notag \\
& &
+\cdots+[\widetilde{I_{A_{1}}}\left(\rho_{A_{N-1}:A_{N}\mid A_{1}}\right)-I_{A_{1}}\left(\rho_{A_{N-1}:A_{N}\mid A_{1}}\right)]\notag \\
&=&\sum_{K=2}^{N-1}[\widetilde{I_{A_{1}}}\left(\rho_{A_{K}:\left(A_{K+1}\cdots A_{N}\right)\mid A_{1}}\right)-I_{A_{1}}\left(\rho_{A_{K}:\left(A_{K+1}\cdots A_{N}\right)\mid A_{1}}\right)]\notag \\
&=&\sum_{K=2}^{N-1}\widetilde{I_{A_{1}}}\left(\rho_{A_{1}A_{K}\left(A_{K+1}\cdots A_{N}\right)}\right)
\nonumber \\
&&-\sum_{K=2}^{N-1}I_{A_{1}}\left(\rho_{A_{1}A_{K}\left(A_{K+1}\cdots A_{N}\right)}\right)_{\left\{ M_{i}^{A_{1}}\right\} }.
\label{disefi123}
\end{eqnarray}
From (\ref{disefi123}), we have $\triangle_{D_{A_{1\left(N\right)}}}^{\rightarrow}\geq 0$ if and only if
$\sum_{K=2}^{N-1}\widetilde{I_{A_{1}}}\left(\rho_{A_{1}A_{K}\left(A_{K+1}\cdots A_{N}\right)}\right)\geq \sum_{K=2}^{N-1}I_{A_{1}}\left(\rho_{A_{1}A_{K}\left(A_{K+1}\cdots A_{N}\right)}\right)_{\left\{ M_{i}^{A_{1}}\right\} }$.

Similarly, we can also get a necessary and sufficient condition for
\begin{eqnarray}
& &\triangle_{D_{A_{1\left(N\right)}}}^{\leftarrow}=\notag \\
& &D^{\leftarrow}\left(\rho_{A_{1}\mid A_{2}\cdots A_{N}}\right)-D^{\leftarrow}\left(\rho_{A_{1}A_{2}}\right)-D^{\leftarrow}\left(\rho_{A_{1}A_{3}}\right)\notag \\
& &-\cdots-D^{\leftarrow}\left(\rho_{A_{1}A_{N}}\right)\notag \\
&=&\sum_{K=2}^{N-1}[\widetilde{I_{A_{1}}}\left(\rho_{A_{1}A_{K}\left(A_{K+1}\cdots A_{N}\right)}\right)-I\left(\rho_{A_{1}A_{K}\left(A_{K+1}\cdots A_{N}\right)}\right)].
\label{disefi1234}
\end{eqnarray}
Where there is no measurement contained in $\widetilde{I_{A_{1}}}\left(\rho_{A_{1}A_{K}\left(A_{K+1}\cdots A_{N}\right)}\right)$, while local measurements
$\{\mathcal{M}_i^{A_m}\}$ ($m=2,\ldots, N$) and coherent measurements $\{\mathcal{M}_i^{(A_k\ldots A_N)}\}$ ($k=2,\ldots, N-1$) contained in
$I\left(\rho_{A_{1}A_{K}\left(A_{K+1}\cdots A_{N}\right)}\right)]$.
\end{proof}

From the above proof, we have $\triangle_{D_{A_{1\left(N\right)}}}^{\leftarrow}\geq0$
if and only if $\sum_{K=2}^{N-1}\widetilde{I_{A_{1}}}\left(\rho_{A_{1}A_{K}\left(A_{K+1}\cdots A_{N}\right)}\right)\geq \sum_{K=2}^{N-1} I\left(\rho_{A_{1}A_{K}\left(A_{K+1}\cdots A_{N}\right)}\right)$.

\section{Summary and discussion}\label{summ}
We have introduced the concept of monogamy deficit by combining together the monogamy inequalities of quantum correlation for multipartite quantum system.
Although two types of monogamy inequalities seem very different on their measurement sides,
based on the concept of monogamy deficit, we have observed a relation
and provided the difference between them. Using this relation, we obtain a unified physical interpretation for these two monogamy deficit.
In addition, we find an interesting fact that there exists a general monogamy condition for
several quantum correlations for tripartite pure states.  By using the concept of
interaction information with respect to one particle, we have proved that the necessary and sufficient condition for
the quantum correlation being monogamous is that
the interrogated interaction information with respect to one particle is less than or equal to the unmeasured interaction information.
Our result can be generalized to $N$-partite system and may have applications in quantum information processing.

\acknowledgements
We thank L. Chen for useful comments.
This work is supported by ``973'' program (2010CB922904) and NSFC (11075126, 11031005, 11175248).

\end{document}